\def\final{1}
\def\icalp{0}

\documentclass[11pt]{article}

\usepackage{algorithm, algorithmic}
\usepackage{amsmath, amssymb, amsthm}
\usepackage{fullpage}
\usepackage{verbatim}
\usepackage{color}
\usepackage{hyperref}

\ifnum\final=0
\newcommand{\mynote}[1]{\marginpar{\tiny\sf #1}}
\else
\newcommand{\mynote}[1]{}
\fi
\newcommand{\Jnote}[1]{\mynote{Jon: {#1}}}

\newcommand\N{\mathbb{N}}
\newcommand\R{\mathbb{R}}
\newcommand\cA{\mathcal{A}}
\newcommand\cE{\mathcal{E}}
\newcommand\cF{\mathcal{F}}
\newcommand\cP{\mathcal{P}}
\newcommand\cQ{\mathcal{Q}}
\newcommand\cR{\mathcal{R}}
\newcommand\cX{\mathcal{X}}

\newcommand\san{\mathcal{A}}

\newcommand\poly{\mathrm{poly}}
\newcommand\bits{\{0,1\}}

\newcommand\set[1]{\left\{#1\right\}} 
\newcommand\eps{\varepsilon}
\newcommand\from{\colon}

\newcommand{\Prob}[2]{\Pr_{#1}\left[ #2 \right]}
\newcommand{\prob}[1]{\Pr\left[#1\right]}
\newcommand{\Ex}[2]{\mathbb{E}_{#1}\left[#2\right]}

\newcommand{\INDSTATE}[1][1]{\STATE\hspace{#1\algorithmicindent}}

\newtheorem{theorem}{Theorem}[section]
\newtheorem{lemma}[theorem]{Lemma}
\newtheorem{fact}[theorem]{Fact}

\newtheorem{remark}[theorem]{Remark}

\theoremstyle{definition}
\newtheorem{definition}[theorem]{Definition}

\title{Faster Algorithms for Privately \\ Releasing Marginals\thanks{An extended abstract of this work appears in \emph{ICALP '12}~\cite{ThalerUlVa12}.} \ifnum\final=0 \\
{\small \sc Working Draft: Please Do Not Distribute}\fi}
\author{Justin Thaler\thanks{\url{http://seas.harvard.edu/\~jthaler}. Supported by the Department of Defense (DoD) through the National
Defense Science \& Engineering Graduate Fellowship (NDSEG) Program,
and in part by NSF grants CCF-0915922 and IIS-0964473.}
 \and Jonathan Ullman\thanks{
\url{http://seas.harvard.edu/\~jullman}.  Supported by NSF grant CNS-0831289 and a gift from Google, Inc.}
 \and Salil Vadhan\thanks{
\url{http://seas.harvard.edu/\~salil}.  Supported by NSF grant CNS-0831289 and a gift from Google, Inc.}
\\ \\ School of Engineering and Applied Sciences \&\\ Center for Research on Computation and Society\\ Harvard University, Cambridge, MA \\ \texttt{\{jthaler,jullman,salil\}@seas.harvard.edu}}
\begin{document}

\maketitle
\begin{abstract}
We study the problem of releasing \emph{$k$-way marginals} of a database $D \in (\bits^d)^n$, while preserving differential privacy.  The answer to a $k$-way marginal query is the fraction of $D$'s records $x \in \bits^d$ with a given value in each of a given set of up to $k$ columns.  Marginal queries enable a rich class of statistical analyses of a dataset, and designing efficient algorithms for privately releasing marginal queries has been identified as an important open problem in private data analysis (cf.~Barak et.~al., PODS '07).

We give an algorithm that runs in time $d^{O(\sqrt{k})}$ and releases a private summary capable of answering any $k$-way marginal query with at most $\pm .01$ error on every query as long as $n \geq d^{O(\sqrt{k})}$.  To our knowledge, ours is the first algorithm capable of privately releasing marginal queries with non-trivial worst-case accuracy guarantees in time substantially smaller than the number of $k$-way marginal queries, which is $d^{\Theta(k)}$ (for $k \ll d$).
\end{abstract}

\section{Introduction}
Consider a database $D \in (\bits^d)^n$ in which each of the $n = |D|$ rows corresponds to an individual's record, and each record consists of $d$ binary attributes.  The goal of privacy-preserving data analysis is to enable rich statistical analyses on the database while protecting the privacy of the individuals.  In this work, we seek to achieve \emph{differential privacy}~\cite{DworkMcNiSm06}, which guarantees that no individual's data has a significant influence on the information released about the database.

One of the most important classes of statistics on a dataset is its \emph{marginals}.  A marginal query is specified by a set $S \subseteq [d]$ and a pattern $t \in \bits^{|S|}$. The query asks, ``What fraction of the individual records in $D$ has each of the attributes $j \in S$ set to $t_j$?''  A major open problem in privacy-preserving data analysis is to \emph{efficiently} create a differentially private summary of the database that enables analysts to answer each of the $3^{d}$ marginal queries.  A natural subclass of marginals are \emph{$k$-way marginals}, the subset of marginals specified by sets $S \subseteq [d]$ such that $|S| \leq k$.

Privately answering marginal queries is a special case of the more general problem of privately answering \emph{counting queries} on the database, which are queries of the form, ``What fraction of individual records in $D$ satisfy some property $q$?''  Early work in differential privacy~\cite{DinurNi03,BlumDwMcNi05,DworkMcNiSm06} showed how to approximately answer any set of of counting queries $\cQ$ by perturbing the answers with appropriately calibrated noise, providing good accuracy (say, within $\pm .01$ of the true answer) as long as $|D| \gtrsim |\cQ|^{1/2}$.

In a setting where the queries arrive online, or are known in advance, it may be reasonable to assume that $|D| \gtrsim |\cQ|^{1/2}$.  However, many situations necessitate a non-interactive data release, where the data owner computes and publishes a single differentially private summary of the database that enables analysts to answer a large class of queries, say all $k$-way marginals for a suitable choice of $k$.  In this case $|\cQ| = d^{\Theta(k)}$, and it may be impractical to collect enough data to ensure $|D| \gtrsim |\cQ|^{1/2}$.  Fortunately, the remarkable work of Blum et.~al.~\cite{BlumLiRo08} and subsequent refinements~\cite{DworkNaReRoVa09,DworkRoVa10,RothRo10,HardtRo10,HardtLiMc10,GuptaRoUl12}, have shown how to privately release approximate answers to any set of counting queries, even when $|\cQ|$ is \emph{exponentially larger} than $|D|$.  For example, these algorithms can release all $k$-way marginals as long as $|D| \geq \tilde{\Theta}(k\sqrt{d})$.  Unfortunately, all of these algorithms have running time at least $2^{d}$, even when $|\cQ|$ is the set of $2$-way marginals (and this is inherent for algorithms that produce ``synthetic data''~\cite{UllmanVa11}; as discussed below).

Given this state of affairs, it is natural to seek \emph{efficient} algorithms capable of privately releasing approximate answers to marginal queries even when $|D| \ll d^k$.  A recent series of works~\cite{GuptaHaRoUl11,CheraghchiKlKoLe12,HardtRoSe12} have shown how to privately release answers to $k$-way marginal queries with small \emph{average error} (over various distributions on the queries) with both running time and minimum database size much smaller than $d^{k}$ (e.g.~$d^{O(1)}$ for product distributions~\cite{GuptaHaRoUl11,CheraghchiKlKoLe12} and $\min\{d^{O(\sqrt{k})}, d^{O(d^{1/3})}\}$ for arbitrary distributions~\cite{HardtRoSe12}).  Hardt et.~al.~\cite{HardtRoSe12} also gave an algorithm for privately releasing $k$-way marginal queries with small \emph{worst-case error} and minimum database size much smaller than $d^{k}$. However the running time of their algorithm is still $d^{\Theta(k)}$, which is polynomial in the number of queries.

In this paper, we give the first algorithms capable of releasing $k$-way marginals up to small worst-case error, with both running time and minimum database size substantially smaller than $d^{k}$.  Specifically, we show how to create a private summary in time $d^{O(\sqrt{k})}$ that gives approximate answers to all $k$-way marginals as long as $|D|$ is at least $d^{O(\sqrt{k})}$.  When $k=d$, our algorithm runs in time $2^{\tilde{O}(\sqrt{d})}$, and is the first algorithm for releasing \emph{all} marginals in time $2^{o(d)}$.

\subsection{Our Results and Techniques}

In this paper, we give faster algorithms for releasing marginals and other classes of counting queries.

\begin{theorem}[Releasing Marginals] \label{thm:main1}
There exists a constant $C$ such that for every $k, d, n \in \N$ with $k \leq d$, every $\alpha \in (0,1]$, and every $\eps > 0$, there is an $\eps$-differentially private sanitizer that, on input a database $D \in (\bits^d)^n$, runs in time $|D| \cdot d^{C\sqrt{k} \log(1/\alpha)}$ and releases a summary that enables computing each of the $k$-way marginal queries on $D$ up to an additive error of at most $\alpha$, provided that $|D| \geq d^{C\sqrt{k} \log(1/\alpha)}/\eps$.\end{theorem}

For notational convenience, we focus on \emph{monotone $k$-way disjunction queries}. However, our results extend straightforwardly to general non-monotone $k$-way disjunction queries (see Section~\ref{sec:monotonedisj}), which are equivalent to $k$-way marginals.  A monotone $k$-way disjunction is specified by a set $S \subseteq [d]$ of size $k$ and asks what fraction of records in $D$ have at least one of the attributes in $S$ set to $1$.  

Our algorithm is inspired by a series of works reducing the problem of private query release to various problems in learning theory.  
One ingredient in this line of work is a shift in perspective introduced by Gupta, Hardt, Roth, and Ullman~\cite{GuptaHaRoUl11}. 
Instead of viewing disjunction queries as a set of functions on the database, they view the database as a function $f_D\from \bits^d \to [0,1]$, in which each vector $s \in \bits^d$ is interpreted as the indicator vector of a set $S \subseteq [d]$, and $f_D(s)$ equals the evaluation of the disjunction specified by $S$ on the database $D$.  They use the structure of the functions $f_D$ to privately learn an approximation $g_D$ that has small \emph{average error} over any product distribution on disjunctions.\footnote{In their learning algorithm, privacy is defined with respect to the rows of the database $D$ that defines $f_{D}$, not with respect to the examples given to the learning algorithm (unlike earlier works on ``private learning''~\cite{KasiviswanathanLeNiRaSm07}). }

Cheraghchi, Klivans, Kothari, and Lee~\cite{CheraghchiKlKoLe12} observed that the functions $f_D$ can be approximated by a low-degree polynomial with small average error over the uniform distribution on disjunctions. They then use a private learning algorithm for low-degree polynomials to release an approximation to $f_D$; and thereby obtain an improved dependence on the accuracy parameter, as compared to~\cite{GuptaHaRoUl11}.  

Hardt, Rothblum, and Servedio~\cite{HardtRoSe12} observe that $f_D$ is itself an average of disjunctions (each row of $D$ specifies a disjunction of bits in the indicator vector $s \in \bits^d$ of the query), and thus develop private learning algorithms for threshold of sums of disjunctions.  These learning algorithms are also based on low-degree approximations of sums of disjunctions.
They show how to use their private learning algorithms to obtain a sanitizer with small average error over \emph{arbitrary distributions} with running time and minimum database size $d^{O(\sqrt{k})}$.  They then are able to apply the private boosting technique of Dwork, Rothblum, and Vadhan~\cite{DworkRoVa10} to obtain worst-case accuracy guarantees.  Unfortunately, the boosting step incurs a blowup of $d^{k}$ in the running time.

We improve the above results by showing how to \emph{directly} compute (a noisy version of) a polynomial $p_D$ that is privacy-preserving and still approximates $f_D$ on \emph{all} $k$-way disjunctions, as long as $|D|$ is sufficiently large. Specifically, the running time and the database size requirement of our algorithm are both polynomial in the number of monomials in $p_D$, which is $d^{O(\sqrt{k})}$.  By ``directly'', we mean that we compute $p_D$ from the database $D$ itself and perturb its coefficients, rather than using a learning algorithm. Our construction of the polynomial $p_D$ uses the same low-degree approximations exploited by Hardt et.~al.~in the development
of their private learning algorithms. 



In summary, the main difference between prior work and ours is that prior work used learning algorithms that have restricted access to the database, and released the hypothesis
output by the learning algorithm. In contrast, we do not make use of any learning algorithms, and give our release algorithm direct access to the database. This enables our algorithm to achieve a worst-case error guarantee while maintaining a minimal database size and running time much smaller than the size of the query set.  Our algorithm is also substantially simpler than that of Hardt et.~al.


We also consider other families of 
counting queries.  We define the class of \emph{$r$-of-$k$ queries}.  Like a monotone $k$-way disjunction, an $r$-of-$k$ query is defined by a set $S \subseteq [d]$ such that $|S| \leq k$.  The query asks what fraction of the rows of $D$ have at least $r$ of the attributes in $S$ set to $1$.  For $r=1$, these queries are exactly monotone $k$-way disjunctions, and $r$-of-$k$ queries are a strict generalization.

\begin{theorem}[Releasing $r$-of-$k$ Queries] \label{thm:main2}
For every $r,k,d, n \in \N$ with $r \leq k \leq d$, every $\alpha \in (0,1]$, and every $\eps > 0$ there is an $\eps$-differentially private sanitizer that, on input a database $D \in (\bits^d)^n$, runs in time $|D| \cdot d^{\tilde{O}\left(\sqrt{rk \log(1/\alpha)}\right)}$ and releases a summary that enables computing each of the $r$-of-$k$ queries on $D$ up to an additive error of at most $\alpha$, provided that $|D| \geq d^{\tilde{O}\left(\sqrt{rk \log(1/\alpha)}\right)}/\eps$.
\end{theorem}
Note that monotone $k$-way disjunctions are just $r$-of-$k$ queries where $r=1$, thus Theorem \ref{thm:main2} 
implies a release algorithm for disjunctions with
quadratically better dependence on $\log(1/\alpha)$, at the cost of slightly worse dependence on $k$ (implicit in the switch from $O(\cdot)$ to $\tilde{O}(\cdot)$).

Finally, we present a sanitizer for privately releasing databases in which
the \emph{rows} of the database are interpreted as decision lists, and the \emph{queries} are inputs to the
decision lists. That is, instead of each record in $D$ being a string of $d$ attributes, each record is an element of the set $\mathrm{DL}_{k,m}$, which consists of all length-$k$ decision lists over $m$ input variables. (See Section~\ref{sec:dls} for a precise definition.) A query is specified by a string $y \in \bits^d$ and asks ``What fraction of database participants would make a certain decision based on the input $y$?''

As an example application, consider a database that allows high school students to
express their preferences for colleges in the form of a decision list.
For example, a student may say, ``If the school is ranked in the top
ten nationwide, I am willing to apply to it. Otherwise, if the school is 
rural, I am unwilling to apply. Otherwise, if the school has a good
basketball team then I am willing to apply to it.'' And so on. Each student is allowed to use up to $k$ attributes out of a set of $m$ binary attributes.  Our sanitizer allows any college (represented by its $m$ binary attributes) to determine the fraction of students willing to apply. 

\begin{theorem}[Releasing Decision Lists] \label{thm:main3}
For any $k, m \in \N$ s.t.~$k \leq m$, any $\alpha \in (0,1]$, and any $\eps > 1/n$, there is an $\eps$-differentially private sanitizer with running time $m^{\tilde{O}\left(\sqrt{k}\log(1/\alpha)\right)}$ that, on input a database $D \in (\mathrm{DL}_{k,m})^n$, releases a summary that enables computing any length-$k$ decision list query up to an additive error of at most $\alpha$ on every query, provided that $|D| \geq m^{\tilde{O}\left(\sqrt{k} \log(1/\alpha)\right)}/\eps$.
\end{theorem}

For comparison, we note that all the results on releasing $k$-way disjunctions (including ours) also apply to a dual setting where the database \emph{records} specify a $k$-way disjunction over $m$ bits and the \emph{queries} are $m$-bit strings (in this setting $m$ plays the role of $d$).  Theorem~\ref{thm:main3} generalizes this dual version of Theorem~\ref{thm:main1}, as length-$k$ decision lists are a strict generalization of $k$-way disjunctions.

We prove the latter two results (Theorems~\ref{thm:main2} and~\ref{thm:main3}) using the same approach outlined for marginals (Theorem~\ref{thm:main1}), but with different low-degree polynomial approximations appropriate for the different types of queries.
\begin{table}
\begin{minipage}[center]{\textwidth}
\begin{center}
\begin{tabular}{|c|c|c|c|c|}
\hline
Paper & Running Time  & Database Size & Error Type\footnote{\emph{Worst case} error indicates that the accuracy guarantee holds for every marginal.  The other types of error indicate that accuracy holds for random marginals over a given distribution from a particular class of distributions (e.g.~product distributions).} & Synthetic Data? \\
\hline
\cite{DinurNi03,DworkNi04,BlumDwMcNi05,DworkMcNiSm06} &$d^{O(k)}$ & $O(d^{k/2} / \alpha)$ & Worst case & N \\
\cite{BarakChDwKaMcTa07}  & $2^{O(d)}$ & $O(d^{k/2} / \alpha)$ & Worst case & Y \\
\cite{BlumLiRo08,DworkNaReRoVa09,DworkRoVa10,HardtLiMc10}& $2^{O(d)}$ & $\tilde{O}(k\sqrt{d} / \alpha^2)$ & Worst case & Y \\
\cite{GuptaHaRoUl11} & $d^{\tilde{O}(1/\alpha^2)}$ & $d^{\tilde{O}(1/\alpha^2)}$ & Product Dists. & N\\
\cite{CheraghchiKlKoLe12} & $d^{O(\log(1/\alpha))}$ & $d^{O(\log(1/\alpha))}$ & Uniform Dist.\footnote{The results of~\cite{CheraghchiKlKoLe12} apply only to the uniform distribution over \emph{all} marginals.} & N \\
\cite{HardtRoSe12} & $d^{O(d^{1/3}\log(1/\alpha))}$ & $d^{O(d^{1/3}\log(1/\alpha))}$ &Any Dist. & N \\
\cite{HardtRoSe12} & $d^{O(k)}$ & $d^{O(d^{1/3}\log(1/\alpha))}$ & Worst case & N \\
\cite{HardtRoSe12} & $d^{O(\sqrt{k}\log(1/\alpha))}$ & $d^{O(\sqrt{k}\log(1/\alpha))}$ & Any Dist. & N \\
\cite{HardtRoSe12} & $d^{O(k)}$ & $d^{O(\sqrt{k}\log(1/\alpha))}$ & Worst case & N  \\
This paper & $d^{O(\sqrt{k}\log(1/\alpha))}$ & $d^{O(\sqrt{k}\log(1/\alpha))}$& Worst case & N \\
\hline
\end{tabular}
\caption{Summary of prior results on differentially private release of $k$-way marginals.  The database size column indicates the minimum database size required to release answers to $k$-way marginals up to an additive error of $\alpha$.  For clarity, we ignore the dependence on the privacy parameters and the failure probability of the algorithms.  Notice that this paper contains the first algorithm capable of releasing $k$-way marginals with running time and worst-case error substantially smaller than the number of queries.}
\end{center}
\end{minipage}
\end{table}
\paragraph{On Synthetic Data.} An attractive type of summary is a \emph{synthetic database}.  A synthetic database is a new database $\widehat{D} \in (\bits^d)^{\widehat{n}}$ whose rows are ``fake'', but such that $\widehat{D}$ approximately preserves many of the statistical properties of the database $D$ (e.g.~all the marginals).  Some of the previous work on counting query release has provided synthetic data, starting with Barak et.~al.~\cite{BarakChDwKaMcTa07} and including~\cite{BlumLiRo08,DworkNaReRoVa09,DworkRoVa10,HardtLiMc10}.

Unfortunately, Ullman and Vadhan~\cite{UllmanVa11} (building on~\cite{DworkNaReRoVa09}) have shown that no differentially private sanitizer with running time $d^{O(1)}$ can take a database $D \in (\bits^d)^n$ and output a private synthetic database $\widehat{D}$, all of whose $2$-way marginals are approximately equal to those of $D$ (assuming the existence of one-way functions).  They also showed that there is a constant $k \in \N$ such that no differentially private sanitizer with running time $2^{d^{1-\Omega(1)}}$ can output a private synthetic database, all of whose $k$-way marginals are approximately equal to those of $D$ (under stronger cryptographic assumptions).

When $k = d$, our sanitizer runs in time $2^{\tilde{O}(\sqrt{d})}$ and releases a private summary that enables an analyst to approximately answer any marginal query on $D$.  Prior to our work it was not known how to release \emph{any} summary enabling approximate answers to all marginals in time $2^{d^{1-\Omega(1)}}$.  Thus, our results show that releasing a private summary for all marginal queries can be done considerably more efficiently if we do not require the summary to be a synthetic database (under the hardness assumptions made in~\cite{UllmanVa11}).
\section{Preliminaries}

\subsection{Differentially Private Sanitizers}\label{sec:sans}

Let a \emph{database} $D \in \cX^n$ be a collection of $n$ rows $x^{(1)}, \dots, x^{(n)}$ from a \emph{data universe} $\cX$.  We say that two databases $D_1,D_2 \in \cX^n$ are \emph{adjacent} if they differ only on a single row, and we denote this by $D_1 \sim D_2$.

A \emph{sanitizer} $\san: \cX^n \to \cR$ takes a database as input and outputs some data structure in $\cR$.  We are interested in sanitizers that satisfy \emph{differential privacy}.
\begin{definition}[Differential Privacy~\cite{DworkMcNiSm06}]\label{def:dp} A sanitizer $\san\from \cX^n \to \cR$ is \emph{$(\eps, \delta)$-differentially private} if for every two adjacent databases $D, D' \in \cX^n$ and every subset $S \subseteq \cR$,
$
\prob{\san(D) \in S} \leq e^{\eps} \prob{\san(D') \in S} + \delta.
$
In the case where $\delta = 0$ we say that $\san$ is \emph{$\eps$-differentially private}.
\end{definition}

Since a sanitizer that always outputs $\bot$ satisfies Definition~\ref{def:dp}, we also need to define what it means for a sanitizer to be accurate.  In particular, we are interested in sanitizers that give accurate answers to \emph{counting queries}.  A counting query is defined by a boolean predicate $q\from \cX \to \bits$.  We define the evaluation of the query $q$ on a database $D \in \cX^n$ to be
$
q(D) = \frac{1}{n} \sum_{i=1}^{n} q(x^{(i)}).
$
We use $\cQ$ to denote a set of counting queries.

Since $\san$ may output an arbitrary data structure, we must specify how to answer queries in $\cQ$ from the output $\san(D)$.  Hence, we require that there is an \emph{evaluator} $\cE\from \cR \times \cQ \to \R$ that estimates $q(D)$ from the output of $\san(D)$.  For example, if $\san$ outputs a vector of ``noisy answers'' $Z = (q(D) + Z_q)_{q \in \cQ}$, where $Z_q$ is a random variable for each $q \in \cQ$, then $\cR = \R^{\cQ}$ and $\cE(Z, q)$ is the $q$-th component of $Z$.  Abusing notation, we write $q(Z)$ and $q(\san(D))$ as shorthand for $\cE(Z, q)$ and $\cE(\san(D),q)$, respectively.  Since we are interested in the efficiency of the sanitization process as a whole, when we refer to the running time of $\san$, we also include the running time of the evaluator $\cE$.  We say that $\san$ is ``accurate'' for the query set $\cQ$ if the values $q(\san(D))$ are close to the answers $q(D)$.  Formally,
\begin{definition}[Accuracy]\label{def:acc}
An output $Z$ of a sanitizer $\san(D)$ is \emph{$\alpha$-accurate} for the query set $\cQ$ if
$
|q(Z) - q(D)| \leq \alpha
$
for every $q \in \cQ$.  A sanitizer is \emph{$(\alpha, \beta)$-accurate} for the query set $\cQ$ if for every database $D$,
$$
\prob{\forall q \in \cQ, |q(\san(D)) - q(D)| \leq \alpha} \geq 1-\beta,
$$
where the probability is taken over the coins of $\san$.
\end{definition}

We will make use of the \emph{Laplace mechanism}.  Let $\mathrm{Lap}^k(\sigma)$ denote a draw from the random variable over $\R^k$ in which each coordinate is chosen independently according to the density function $\mathrm{Lap}_{\sigma}(x) \propto e^{-|x|/\sigma}$.  Let $D \in \cX^n$ be a database and $g: \cX^n \to \R^k$ be a function such that for every pair of adjacent databases $D \sim D'$, $\| g(D) - g(D') \|_\infty \leq \Delta.$  Then we have the following two theorems:
\begin{lemma}[Laplace Mechanism, $\eps$-Differential Privacy~\cite{DworkMcNiSm06}] \label{thm:laplacemech}
For $D, g, k, \Delta$ as above, the mechanism $\cA(D) = g(D) + \mathrm{Lap}^k(\Delta k/\eps)$ satisfies $\eps$-differential privacy.  Furthermore, for any $\beta > 0$,
$\Prob{\san}{ \| g(D) - \cA(D) \|_1 \leq \alpha} \geq 1-\beta,$
for
$\alpha = 2 \Delta k^2 \log(k/\beta) / \eps.$
\end{lemma}
The choice of the $L_1$ norm in the accuracy guarantee of the lemma is for convenience, and doesn't matter for the parameters of Theorems~\ref{thm:main1}-\ref{thm:main3} (except for the hidden constants).

\ifnum\icalp=0
If the privacy requirement is relaxed to $(\eps,\delta)$-differential privacy (for $\delta > 0)$, then it is sufficient to perturb each coordinate of $g(D)$ with noise from a Laplace distribution of smaller magnitude, leading to smaller error. 
\begin{lemma}[Laplace Mechanism, $(\eps, \delta)$-Differential Privacy~\cite{DinurNi03,DworkNi04,BlumDwMcNi05,DworkRoVa10}] \label{thm:advcomp}
For $D, g, k, \Delta$ as above, and for every $\delta > 0$, the mechanism $\cA(D) = g(D) + \mathrm{Lap}^k(3\Delta \sqrt{k\log(1/\delta)}/\eps)$ satisfies $(\eps, \delta)$-differential privacy.
Furthermore, for any $\beta > 0$,
$\Prob{\san}{ \| g(D) - \cA(D) \|_1 \leq \alpha} \geq 1-\beta,$
for
$\alpha = 6\Delta k \sqrt{k\log(1/\delta)} \log(k/\beta) / \eps.$
\end{lemma}
\else

\fi

\subsection{Query Function Families} \label{sec:countingqueries}

We take the approach of Gupta et.~al.~\cite{GuptaHaRoUl11} and think of the database $D$ as specifying a function $f_D$ mapping queries $q$ to their answers $q(D)$, which we call the \emph{$\cQ$-representation of $D$}. We now describe this transformation more formally:

\begin{definition}[$\cQ$-Function Family] \label{def:Qfunction}
Let $\cQ = \set{q_y}_{y \in Y_{\cQ} \subseteq \bits^m}$ be a set of counting queries on a data universe $\cX$, where each query is indexed by an $m$-bit string.  We define the \emph{index set of $\cQ$} to be the set $Y_{\cQ} = \set{y \in \bits^m \mid q_y \in \cQ}$. 

We define the \emph{$\cQ$-function family} $\cF_{\cQ} = \set{ f_{x}: \bits^m \to \bits}_{x \in \cX}$ as follows:
For every possible database row $x \in \cX$, the function $f_{\cQ, x}: \bits^m \to \bits$ is defined as $f_{\cQ, x}(y) = q_y(x)$.  
Given a database $D \in \cX^n$ we define the function $f_{\cQ,D}: \bits^m \to [0,1]$ where $f_{\cQ,D}(q) = \frac{1}{n} \sum_{i=1}^{n} f_{\cQ, x^{(i)}}(q)$.  When $\cQ$ is clear from context we will drop the subscript $\cQ$ and simply write $f_x$, $f_D$, and $\cF$.
\end{definition}

For some intuition about this transformation, when the queries are monotone $k$-way disjunctions on a database $D \in (\bits^d)^n$, the queries are defined by sets $S \subseteq [d]$ , $|S| \leq k$.  In this case each query can be represented by the $d$-bit indicator vector of the set $S$, with at most $k$ non-zero entries.  Thus we can take $m = d$ and $Y_{\cQ} = \set{ y \in \bits^d \mid \sum_{j=1}^{d} y_j \leq k}$.
\subsection{Polynomial Approximations} \label{sec:polys}

\Jnote{I changed the definition of norm to be $L_{\infty}$, it seemed more natural, and more importantly makes it easier to state an $(\eps,\delta)$ version}
An $m$-variate real polynomial $p \in \R[y_1, \dots, y_m]$ of \emph{degree $t$} and \emph{($L_{\infty}$) norm $T$} can be written as
$
p(y) = \sum_{j_1, \dots, j_m \geq 0 \atop j_1 + \dots + j_m \leq t} c_{j_1, \dots, j_m} \prod_{\ell=1}^{m} y_{\ell}^{j_\ell}
$
where
$
|c_{j_1, \dots, j_m}| \leq T
$
for every $j_1, \dots, j_m$.
Recall that there are at most $\binom{m+t}{t}$ coefficients in an $m$-variate polynomial of total degree $t$. 
Often we will want to associate a polynomial $p$ of degree $t$ and norm $T$ with its coefficient vector $\vec{p} \in [-T,T]^{\binom{m+t}{t}}$.  Specifically,
$\vec{p} = (c_{j_1, \dots, j_m})_{j_1, \dots, j_m \geq 0 \atop j_1 + \dots + j_m \leq t}.$
Given a vector $\vec{p}$ and a point $y \in \bits^{m}$ we use $\vec{p}(y)$ to indicate the evaluation of the polynomial described by the vector $\vec{p}$ at the point $y$.  Observe this is equivalent to computing $\vec{p} \cdot \vec{y}$ where $\vec{y} \in \bits^{\binom{m+t}{t}}$ is defined as $y_{j_1, \dots, j_m} = \prod_{\ell=1}^{m} y_{\ell}^{j_{\ell}}$ for every $j_1, \dots, j_m \geq 0$, $j_1 + \dots + j_m \leq t$.

Let $\cP_{t,T}$ be the family of all $m$-variate real polynomials of degree $t$ and norm $T$.  In many cases, the functions $f_{\cQ, x}: \bits^m \to \bits$ can be approximated well on all the indices in $Y_{\cQ}$ by a family of polynomials $\cP_{t,T}$ with low degree and small norm.  Formally:
\begin{definition}[Uniform Approximation by Polynomials] \label{def:approxbypolys}
Given a family of $m$-variate functions $\cF = \set{f_x}_{x \in \cX}$ and a set $Y \subseteq \bits^m$, we say that the family $\cP_{t,T}$ \emph{uniformly $\gamma$-approximates $\cF$ on $Y$} if for every $x \in \cX$, there exists $p_{x} \in \cP_{t,T}$ such that $\max_{y \in Y} | f_x(y) - p_{x}(y) | \leq \gamma$.

We say that $\cP_{t,T}$ \emph{efficiently and uniformly $\gamma$-approximates $\cF$} if there is an algorithm $\cP_{\cF}$ that takes $x \in \cX$ as input, runs in time $\poly(\log|\cX|, \binom{m+t}{t}, \log T)$, and outputs a coefficient vector $\vec{p}_{x}$ such that $\max_{y \in Y} |f_x(y) - \vec{p}_{x}(y)| \leq \gamma$.
\end{definition}
\section{From Polynomial Approximations to Data Release Algorithms} \label{sec:release}
In this section we present an algorithm for privately releasing any family of counting queries $\cQ$ such that $\cF_{\cQ}$ that can be efficiently and uniformly approximated by polynomials.  The algorithm will take an $n$-row database $D$ and, for each row $x \in D$, constructs a polynomial $p_{x}$ that uniformly approximates the function $f_{\cQ, x}$ (recall that $f_{\cQ, x}(q) = q(x)$, for each $q \in \cQ$).  From these, it constructs a polynomial $p_{D} = \frac{1}{n} \sum_{x \in D} p_{x}$ that uniformly approximates $f_{\cQ, D}$.  The final step is to perturb each of the coefficients of $p_{D}$ using noise from a Laplace distribution (Theorem~\ref{thm:laplacemech}) and bound the error introduced from the perturbation.

\begin{theorem}[Releasing Polynomials] \label{thm:releasepolys}
Let $\cQ = \set{q_y}_{y \in Y_{\cQ} \subseteq \bits^m}$ be a set of counting queries over $\bits^d$, and $\cF_{\cQ}$ be the $\cQ$ function family (Definition~\ref{def:Qfunction}).  Assume that $\cP_{t,T}$ efficiently and uniformly $\gamma$-approximates $\cF_{\cQ}$ on $Y_{\cQ}$ (Definition~\ref{def:approxbypolys}).  Then there is a sanitizer $\san\from (\bits^d)^n \to \R^{\binom{m+t}{t}}$ that
\begin{enumerate}
\item is $\eps$-differentially private,
\item runs in time $\poly(n, d, \binom{m+t}{t}, \log T, \log(1/\eps))$, and
\item is $(\alpha, \beta)$-accurate for $\cQ$ for
$
\alpha = \gamma + \frac{4 T \binom{m+t}{t}^2 \log\left(\binom{m+t}{t} / \beta\right)}{\eps n}.
$
\end{enumerate}
\end{theorem}

\begin{proof}
First we construct the sanitizer $\cA$. See the relevant codebox below.

\begin{algorithm}
The Sanitizer $\cA$
\begin{algorithmic}
\STATE{\textbf{Input}: A database $D \in (\bits^d)^n$, an explicit family of polynomials $\cP$, and a parameter $\eps > 0$.}
\STATE{\textbf{For} $i = 1, \dots, n$}
\INDSTATE[1]{Using efficient approximation of $\cF$ by $\cP$, compute a polynomial $\vec{p}_{x^{(i)}} = \cP_{\cF}(x^{(i)})$ that $\gamma$-approximates $f_{x^{(i)}}$ on $Y_{\cQ}$.}
\STATE{\textbf{Let} $\vec{p}_{D} = \frac{1}{n} \sum_{i=1}^{n} \vec{p}_{x^{(i)}}$, where the sum denotes standard entry-wise vector addition.}
\STATE{\textbf{Let} $\tilde{\vec{p}}_{D} = \vec{p}_D + Z$, where $Z$ is drawn from an $\binom{m+t}{t}$-variate Laplace distribution with parameter $2T/\eps n$ (Section~\ref{sec:sans}).}
\STATE{\textbf{Output}: $\tilde{p}_{D}$.}
\end{algorithmic}
\label{codebox1}
\end{algorithm}
\paragraph{Privacy.}
We establish that $\cA$ is $\eps$-differentially private.  This follows from the observation that for any two adjacent $D \sim D'$ that differ only on row $i^*$,
\begin{align*}
\| \vec{p}_{D} - \vec{p}_{D'} \|_\infty
&={} \left\| \frac{1}{n} \sum_{i =1}^{n} \vec{p}_{x^{(i)}} - \frac{1}{n} \sum_{i=1}^{n} \vec{p}_{x'^{(i)}} \right\|_\infty
={} \frac{1}{n} \left\| \vec{p}_{x^{(i^*)}} - \vec{p}_{x'^{(i^*)}} \right\|_\infty \leq \frac{2T}{n}.
\end{align*}
The last inequality is from the fact that for every $x$, $\vec{p}_x$ is a vector of $L_\infty$ norm at most $T$.  Part~1 of the Theorem now follows directly from the properties of the Laplace Mechanism (Theorem~\ref{thm:laplacemech}).
Now we construct the evaluator $\cE$.
\begin{algorithm}
The Evaluator $\cE$ for the Sanitizer $\cA$
\begin{algorithmic}
\STATE{\textbf{Input}: A vector $\tilde{\vec{p}} \in \R^{\binom{m+t}{t}}$ and the description of a query $y \in \bits^m$.}
\STATE{\textbf{Output}: $\tilde{\vec{p}}(y)$.  Recall that we view $\tilde{\vec{p}}$ as an $m$-variate polynomial, $p$, and $\tilde{\vec{p}}(y)$ is the evaluation of $p$ on the point $y$.}
\end{algorithmic}
\end{algorithm}

\paragraph{Efficiency.}
Next, we show that $\cA$ runs in time $\poly(n, d,\binom{m+t}{t}, \log T, \log(1/\eps))$.  Recall that we assumed the polynomial construction algorithm $\cP$ runs in time $\poly(d, \binom{m+t}{t}, \log T)$.  The algorithm $\cA$ needs to run $\cP_{\cF}$ on each of the $n$ rows, and then it needs to generate $\binom{m+t}{t}$ samples from a univariate Laplace distribution with magnitude $\poly(T, \binom{m+t}{t}, 1/n, 1/\eps)$, which can also be done in time $\poly(\binom{m+t}{t}, \log T, \log n, \log(1/\eps))$.  We also establish that $\cE$ runs in time $\poly(\binom{m+t}{t}, \log T, \log n, \log(1/\eps))$, observe that $\cE$ needs to expand the input into an appropriate vector of dimension $\binom{m+t}{t}$ and take the inner product with the vector $\tilde{\vec{p}}$, whose entries have magnitude $\poly(\binom{m+t}{t}, T, 1/n, 1/\eps)$.  These observations establish Part 2 of the Theorem.

\medskip

\noindent
\textit{Accuracy.}
Finally, we analyze the accuracy of the sanitizer $\cA$.  First, by the assumption that $\cP_{t, T}$ uniformly $\gamma$-approximates $\cF$ on $Y \subseteq \bits^m$, we have
\ifnum\icalp=0
\begin{align*}
\max_{y \in Y} \left| f_{D}(y) - \vec{p}_{D}(y) \right| 
&={} \max_{y \in Y} \left| \frac{1}{n} \sum_{i=1}^{n} f_{x^{(i)}}(y) - \frac{1}{n} \sum_{i=1}^{n} \vec{p}_{x^{(i)}}(y) \right| \\
&\leq{} \frac{1}{n} \sum_{i=1}^{n} \max_{y \in Y} \left| f_{x^{(i)}}(y) - \vec{p}_{x^{(i)}}(y) \right|
\leq{} \gamma.
\end{align*}
\else
\begin{align*}
\max_{y \in Y} \left| f_{D}(y) - \vec{p}_{D}(y) \right| 
&\leq{} \frac{1}{n} \sum_{i=1}^{n} \max_{y \in Y} \left| f_{x^{(i)}}(y) - \vec{p}_{x^{(i)}}(y) \right|
\leq{} \gamma.
\end{align*}
\fi
Now we want to establish that 
$
\prob{\max_{y \in \bits^m} \left| \tilde{\vec{p}}_{D}(y) - \vec{p}_{D}(y) \right| \leq \alpha'} \geq 1-\beta
$
for $$\alpha' = \frac{4T\binom{m+t}{t}^2 \log\left( \binom{m+t}{t} /\beta \right) }{ \eps n},$$ where the probability is taken over the coins of $\cA$.  Part (3) of the Theorem will then follow by the triangle inequality.

To see that the above statement is true, observe that by the properties of the Laplace mechanism (Theorem~\ref{thm:laplacemech}), we have
$
\prob{ \left\| \tilde{\vec{p}}_{D} - \vec{p}_{D} \right\|_1 \leq \alpha'} \geq 1-\beta,
$
where the probability is taken over the coins of $\cA$. Given that $\|\tilde{\vec{p}}_{D} - \vec{p}_{D} \|_1 \leq \alpha'$, it holds that for every $y \in \bits^m$,
\begin{align*}
\left| \tilde{\vec{p}}_{D}(y) - \vec{p}_D(y) \right| 
&={} \left| (\tilde{\vec{p}}_{D} - \vec{p}_D)(y) \right| \leq{} \| \tilde{\vec{p}}_{D} - \vec{p}_D \|_1
\leq{} \alpha'.
\end{align*}
The first inequality follows from the fact that every monomial evaluates to $0$ or $1$ at the point $y$.  This completes the proof of the theorem.

\end{proof}

Using Theorem~\ref{thm:advcomp}, we can improve the bound on the error at the expense of relaxing the privacy guarantee to $(\eps, \delta)$-differential privacy.  This improved error only affects the hidden constants in Theorems~\ref{thm:main1}-\ref{thm:main3}, so we only state those theorems for $\eps$-differential privacy.
\begin{theorem}[Releasing Polynomials, $(\eps, \delta)$-Differential Privacy] \label{thm:advreleasepolys}
Let $\cQ = \set{q_y}_{y \in Y_{\cQ} \subseteq \bits^m}$ be a set of counting queries over $\bits^d$, and $\cF_{\cQ}$ be the $\cQ$ function family (Definition~\ref{def:Qfunction}).  Assume that $\cP_{t,T}$ efficiently and uniformly $\gamma$-approximates $\cF_{\cQ}$ on $Y_{\cQ}$ (Definition~\ref{def:approxbypolys}).  Then there is a sanitizer $\san\from (\bits^d)^n \to \R^{\binom{m+t}{t}}$ that
\begin{enumerate}
\item is $(\eps,\delta)$-differentially private,
\item runs in time $\poly(n, d, \binom{m+t}{t}, \log T, \log(1/\eps), \log (1/\delta))$,
\item is $(\alpha, \beta)$-accurate for $\cQ$ for
$
\alpha = \gamma + \frac{12 T \binom{m+t}{t} \sqrt{\binom{m+t}{t} \log(1/\delta)} \log\left(\binom{m+t}{t} / \beta\right)}{\eps n}.
$
\end{enumerate}
\end{theorem}
The proof of this theorem is identical to that of Theorem~\ref{thm:releasepolys}, but using the analysis of the Laplace mechanism from Theorem~\ref{thm:advcomp} in place of that of Theorem~\ref{thm:laplacemech}.
\section{Applications} \label{sec:apps}
In this section we establish the existence of explicit families of low-degree polynomials approximating the families $\cF_{\cQ}$ for some interesting query sets.
\subsection{Releasing Monotone Disjunctions} \label{sec:monotonedisj}

We define the class of monotone $k$-way disjunctions as follows:

\begin{definition}[Monotone $k$-Way Disjunctions] \label{def:monotonedisj}Let $\cX = \bits^d$.
The query set $\cQ_{\mathrm{Disj}, k} = \set{q_y}_{y \in Y_k \subseteq \bits^d}$ of \emph{monotone $k$-way disjunctions over $\bits^d$} contains a query $q_y$ for every $y \in Y_k = \set{y \in \bits^d \mid |y| \leq k}$.  Each query is defined as $q_y(x_1, \dots, x_d) = \bigvee_{j =1}^{d} y_j x_j$.
The $\cQ_{\mathrm{Disj}, k}$ function family $\cF_{\cQ_{\mathrm{Disj}, k}} = \set{f_x}_{x \in \bits^d}$ contains a function
$f_x(y_1, \dots, y_d) = \bigvee_{j=1}^{d} y_j x_j$
for every $x \in \bits^d$.
\end{definition}

\ifnum\icalp=0
Thus the family $\cF_{\cQ_{\mathrm{Disj}, k}}$ consists of all disjunctions, and the image of $\cQ_{\mathrm{Disj}, k}$, which we denote $Y_k$, consists of all vectors $y \in \bits^d$ with at most $k$ non-zero entries.  We can approximate disjunctions over the set $Y_k$ using a well-known transformation of the Chebyshev polynomials (see, e.g., \cite[Theorem 8]{KlivansSe04} and \cite[Claim 5.4]{HardtRoSe12}).  First we recall the useful properties of the univariate Chebyshev polynomials.

\begin{fact} [Chebyshev Polynomials] \label{lemma:chebyshev}
For every $k \in \N$ and $\gamma > 0$, there exists a univariate real polynomial $g_k(x) = \sum_{i=0}^{t_k} c_i x^i$ of degree $t_k$ such that 
\ifnum\icalp=0
\begin{enumerate}
\item $t_k = O(\sqrt{k} \log(1/\gamma))$,
\item for every $i \in \set{0,1,\dots,t_k}, |c_i| \leq  2^{O(\sqrt{k} \log(1/\gamma))}$,
\item $g_k(0) = 0$, and
\item for every $x \in \set{1, \dots, k}$, $1-\gamma \leq g_k(x) \leq 1+\gamma$.
\end{enumerate}
\else
(1) $t_k = O(\sqrt{k} \log(1/\gamma))$, (2) for every $i \in \set{0,1,\dots,t_k}, |c_i| \leq  2^{O(\sqrt{k} \log(1/\gamma))}$, (3) $g_k(0) = 0$, and (4) for every $x \in \set{1, \dots, k}$, $1-\gamma \leq g_k(x) \leq 1+\gamma$.
\fi
Moreover, such a polynomial can be constructed in time $\poly(k, \log(1/\gamma))$ (e.g.~using linear programming, though more efficient algorithms are known).
\end{fact}

We can use Lemma \ref{lemma:chebyshev} to approximate $k$-way monotone disjunctions. Note that our result easily extends to monotone $k$-way conjunctions via the identity \\ $\wedge_{j=1}^{d} x_j y_j = 1 - \vee_{j=1}^{d} (1-x_j) y_j$. 
Moreover, it extends to non-monotone conjunctions and disjunctions: we may extend the data universe
as in \cite[Theorem 1.2]{HardtRoSe12} to $\{0, 1\}^{2d}$, and include the negation of each item in the original domain. Non-monotone conjunctions over domain $\{0, 1\}^d$ correspond to monotone conjunctions over the expanded domain $\{0,1\}^{2d}$.

The next lemma shows that $\cF_{\cQ_{\mathrm{Disj}, k}}$ can be efficiently and uniformly approximated by polynomials of low degree and low norm.  The statement is a well-known application of Chebyshev polynomials, and a similar statement appears in~\cite{HardtRoSe12} but without bounding the running time of the construction or a bound on the norm of the polynomials.  We include the statement and a proof for completeness, and to verify the additional properties we need. 

\begin{lemma}[Approximating $\cF_{\cQ_{\mathrm{Disj}, k}}$ by polynomials, similar to~\cite{HardtRoSe12}] \label{thm:approxdisj}
For every $k,d \in \N$ such that $k \leq d$ and every $\gamma > 0$, the family $\cP_{t, T}$ of $d$-variate real polynomials of degree $t = O(\sqrt{k} \log(1/\gamma))$ and norm $T = d^{O(\sqrt{k} \log(1/\gamma))}$ efficiently and uniformly $\gamma$-approximates the family $\cF_{\cQ_{\mathrm{Disj}, k}}$ on the set $Y_k$.
\end{lemma}
\begin{proof}
The algorithm $\cP_{\mathrm{Disj}, k}$ for constructing the polynomials appears in the relevant codebox above.
\begin{algorithm}
$\cP_{\mathrm{Disj},k}$
\begin{algorithmic}
\STATE{\textbf{Input:} a vector $x \in \bits^d$.}
\STATE{\textbf{Let:} $g_k$ be the polynomial described in Lemma~\ref{lemma:chebyshev}.}
\STATE{\textbf{Let:} $\vec{p}_x \in \R^{\binom{m+t_k}{t_k}}$ be the expansion of $p_x(y_1, \dots, y_d) = g_k\left(\sum_{j=1}^{d} y_j x_j\right).$}
\STATE{\textbf{Output:} $\vec{p}_x$.}
\end{algorithmic}
\end{algorithm}

Since $p_{x}$ is a degree-$t_k$ polynomial applied to a degree-1 polynomial (in the variables $y_j$), its degree is at most $t_k$.  To see the stated norm bound, note that every monomial of total degree $i$ in $p_{x}$ comes from the expansion of $\left(\sum_{j=1}^{d} y_jx_j\right)^{i}$, and every coefficient in this expansion is a non-negative integer less than or equal to $k^i$. In $p_{x}$, each of these terms is multiplied by $c_i$ (the $i$-th coefficient of $g_k$).  Thus the norm of $p_{x}$ is at most $ \max_{i \in \set{0,1,\dots,t_k}} k^i \cdot |c_i| = k^{O(\sqrt{k} \log(1/\gamma))}= d^{O(\sqrt{k} \log(1/\gamma))}  $.  To see that $\cP_{\mathrm{Disj}, k}$ is efficient, note that we can find every coefficient of $p_{x}$ of total degree $i$ by expanding $\left(\sum_{j=1}^{d} y_jx_j\right)^{i}$ into all of its $d^i$ terms and multiplying by $c_i$, which can be done in time $\poly(d^{t_k}) = \poly(\binom{d+t_k}{t_k})$, as is required.

To see that $\cP_{\mathrm{Disj}, k}$ $\gamma$-approximates $\cF_{\cQ_{\mathrm{Disj}, k}}$, observe that for every $x, y \in \bits^d$,
$
f_x(y) = 0 \Rightarrow p_x(y) =  0,
$
and for every $x \in \bits^d, y \in Y_k$,
$
f_x(y) = 1 \Rightarrow 1 - \gamma \leq p_x(y) \leq 1 + \gamma.
$
This completes the proof.
\end{proof}
\else
Thus the family $\cF_{\cQ_{\mathrm{Disj}, k}}$ consists of all disjunctions, and the index set, $Y_k$, consists of all vectors $y \in \bits^d$ with at most $k$ non-zero entries.  

The next lemma shows that $\cF_{\cQ_{\mathrm{Disj}, k}}$ can be efficiently and uniformly approximated by polynomials of low degree and low norm.  The statement is a well-known application of Chebyshev polynomials, and a similar statement appears in~\cite{HardtRoSe12} but without bounding the running time of the construction or a bound on the norm of the polynomials. 
\begin{lemma}[Approximating $\cF_{\cQ_{\mathrm{Disj}, k}}$ by polynomials, similar to~\cite{HardtRoSe12}] \label{thm:approxdisj}
For every $k,d \in \N$ such that $k \leq d$ and every $\gamma > 0$, the family $\cP_{t, T}$ of $d$-variate real polynomials of degree $t = O(\sqrt{k} \log(1/\gamma))$ and norm $T = d^{O(\sqrt{k} \log(1/\gamma))}$ efficiently and uniformly $\gamma$-approximates the family $\cF_{\cQ_{\mathrm{Disj}, k}}$ on the set $Y_k$.
\end{lemma}
\fi

Theorem~\ref{thm:main1} in the introduction follows by combining Theorems~\ref{thm:releasepolys} and~\ref{thm:approxdisj}.

\subsection{Releasing Monotone $r$-of-$k$ Queries}
We define the class of monotone $r$-of-$k$ queries as follows:

\begin{definition}[Monotone $r$-of-$k$ Queries] Let $\cX = \bits^d$ and $r,k \in \N$ such that $r \leq k \leq d$.
The query set $\cQ_{r, k} = \set{q_y}_{y \in Y_k \subseteq \bits^d}$ of \emph{monotone $r$-of-$k$ queries over $\bits^d$} contains a query $q_y$ for every $y \in Y_k = \set{y \in \bits^d \mid |y| \leq k}$.  Each query is defined as $q_y(x_1, \dots, x_d) = \mathbf{1}_{\sum_{j=1}^{d} y_j x_j \geq r}$.
The $\cQ_{r, k}$ function family $\cF_{\cQ_{r, k}} = \set{f_x}_{x \in \bits^d}$ contains a function
$f_x(y_1, \dots, y_d) = \mathbf{1}_{\sum_{j=1}^{d} y_j x_j \geq r}$
for every $x \in \bits^d$.
\end{definition}

\ifnum\icalp=0
Sherstov \cite[Lemma 3.11]{Sherstov09} 
gives an explicit construction of polynomials that can be used to approximate the family $\cF_{\cQ_{r,k}}$ over $Y_k$ with low degree.
It can be verified by inspecting the construction that the coefficients of the resulting polynomial are not too large.

\begin{lemma}[\cite{Sherstov09}] \label{lemma:paturi}
For every $r,k \in \N$ such that $r \leq k$ and $\gamma > 0$, there exists a univariate polynomial $g_{r,k}\from \R \to \R$ of degree $t_{r,k}$ such that $g_{r,k}(x) = \sum_{i=0}^{t_k} c_i x^i$ and
\ifnum\icalp=0
\begin{enumerate}
\item $t_{r,k} = O\left(\sqrt{rk} \log(k) + \sqrt{k \log(1/\gamma)\log(k)}\right)$,
\item for every $i \in \set{0,1,\dots, t_k}, |c_i| \leq  2^{\tilde{O}(\sqrt{kr \log(1/\gamma)})}$,
\item for every $x \in \set{0,1,\dots,r-1}$, $-\gamma \leq g_{r,k}(x) \leq \gamma$, and
\item for every $x \in \set{r, \dots, k}$, $1-\gamma \leq g_{r,k}(x) \leq 1+\gamma$.
\end{enumerate}
\else
(1) $t_{r,k} = O\left(\sqrt{rk} \log(k) + \sqrt{k \log(1/\gamma)\log(k)}\right)$,
(2) for every $i \in \set{0,1,\dots, t_k}, |c_i| \leq  2^{\tilde{O}(\sqrt{kr \log(1/\gamma)})}$,
(3) for every $x \in \set{0,1,\dots,r-1}$, $-\gamma \leq g_{r,k}(x) \leq \gamma$, and
(4) for every $x \in \set{r, \dots, k}$, $1-\gamma \leq g_{r,k}(x) \leq 1+\gamma$.
\fi
Moreover, $g_{r,k}$ can be constructed in time $\poly(k, r, \log(1/\gamma))$ (e.g.~using linear programming).
\end{lemma}
For completeness we include a proof of Lemma~\ref{lemma:paturi} in the appendix.
We can use these polynomials to approximate monotone $r$-of-$k$ queries.

\begin{lemma}[Approximating $\cF_{\cQ_{r, k}}$ on $Y_k$] \label{thm:approx2}
For every $r,k,d \in \N$ such that $r \leq k \leq d$ and every $\gamma > 0$, the family $\cP_{t,T}$ of $d$-variate real polynomials of degree $t = \tilde{O}(\sqrt{kr \log(1/\gamma)})$ and norm $T = d^{\tilde{O}(\sqrt{kr \log(1/\gamma)})}$ efficiently and uniformly $\gamma$-approximates the family $\cF_{\cQ_{r, k}}$ on the set $Y_k$.
\end{lemma}
\begin{proof}
The construction and proof is identical to that of Theorem~\ref{thm:approxdisj} with the polynomials of Lemma~\ref{lemma:paturi} in place of the polynomials described in Lemma \ref{lemma:chebyshev}.
\end{proof}

Theorem~\ref{thm:main2} in the introduction now follows by combining Theorems~\ref{thm:releasepolys} and~\ref{thm:approx2}. Note that our result also extends easily
to non-monotone $r$-of-$k$ queries in the same manner as Theorem \ref{thm:main1}.  
\else
The next lemma shows that $\cF_{\cQ_{r,k}}$ can be efficiently and uniformly approximated over $Y_k$ by low-degree polynomials.  The statement is based on approximation-theoretic results of Sherstov \cite[Lemma 3.11]{Sherstov09}.
\begin{lemma}[Approximating $\cF_{\cQ_{r, k}}$ on $Y_k$] \label{thm:approx2}
For every $r,k,d \in \N$ such that $r \leq k \leq d$ and every $\gamma > 0$, the family $\cP_{t,T}$ of $d$-variate real polynomials of degree $t = \tilde{O}(\sqrt{kr \log(1/\gamma)})$ and norm $T = d^{\tilde{O}(\sqrt{kr \log(1/\gamma)})}$ efficiently and uniformly $\gamma$-approximates the family $\cF_{\cQ_{r, k}}$ on the set $Y_k$.
\end{lemma}
\fi

\begin{remark}
Using the principle of inclusion-exclusion, the answer to a monotone $r$-of-$k$ query can be written as a linear combination of the answers to $k^{O(r)}$ monotone $k$-way disjunctions. Thus, a sanitizer that is $(\alpha/k^{O(r)}, \beta)$-accurate for monotone $k$-way disjunctions implies a sanitizer that is $(\alpha, \beta)$-accurate for monotone $r$-of-$k$ queries. However, combining this implication with Theorem~\ref{thm:main1} yields a sanitizer with running time $d^{O(r \sqrt{k} \log(k/\beta))}$, which has a worse dependence on $r$ than what we achieve in Theorem~\ref{thm:main2}.
\end{remark}

\subsection{Releasing Decision Lists} \label{sec:dls}

A \emph{length-$k$ decision list} over $m$ variables is a function which can be written in the form
``if $\ell_1$ then output $b_1$ else $\cdots$ else if $\ell_k$
then output $b_k$ else output $b_{k+1}$,'' where each
$\ell_i$ is a boolean literal in $\{x_1, \dots, x_m\}$, and each $b_i$ is an output
bit in $\{0,1\}.$ Note that decision lists of length-$k$ strictly generalize $k$-way disjunctions and conjunctions.  We use $\mathrm{DL}_{k,m}$ to denote the set of all length-$k$ decision lists over $m$ binary input variables.

\begin{definition}[Evaluations of Length-$k$ Decision Lists] Let $k, m \in \N$ such that $k \leq m$ and $\cX = \mathrm{DL}_{k,m}$.
The query set $\cQ_{\mathrm{DL}_{k,m}} = \set{q_y}_{y \in \bits^m}$ of \emph{evaluations of length-$k$ decision lists} contains a query $q_y$ for every $y \in \bits^m$.  Each query is defined as $q_y(x) = x(y)$ where $x \in \mathrm{DL}_{k,m}$ is a length-$k$ decision list over $m$ variables.
The $\cQ_{\mathrm{DL}_{k,m}}$ function family $\cF_{\cQ_{\mathrm{DL}_{k,m}}} = \set{f_x}_{x \in \mathrm{DL}_{k,m}}$ contains functions $f_x(y) = x(y)$ for every $x \in \mathrm{DL}_{k,m}$.  That is, $\cF_{\cQ_{\mathrm{DL}_{k,m}}} = \mathrm{DL}_{k,m}$.
\end{definition}

We clarify that in this setting, the records in the database are length-$k$ decision lists over $\bits^m$ and the queries inputs in $\bits^m$.  Thus $|\cX| = |\mathrm{DL}_{k,m}| = m^{O(k)}$ and $|\cQ| = 2^m$.  Alternatively, $\cX = \bits^d$ for $d = k(\log m + 2) + 1$, since a length-$k$ decision list can be described using this many bits.
\ifnum\icalp=0
Klivans and Servedio \cite[Claim 5.4]{KlivansSe04} have shown that decision lists of length $k$ can be uniformly approximated to accuracy $\gamma$ by low-degree polynomials.  We give a self-contained proof of this fact in Appendix \ref{app:dlfact} for completeness.
\else
\fi

\begin{lemma}[\cite{KlivansSe04}] \label{dlfact} For every $k, m \in \N$ such that $k \leq m$ and every $\gamma > 0$, the family $\cP_{t, T}$ of $m$-variate real polynomials of degree $\tilde{O}\left(\sqrt{k} \log(1/\gamma)\right)$ and norm $T = m^{\tilde{O}\left(\sqrt{k} \log(1/\gamma)\right)}$ efficiently and uniformly $\gamma$-approximates the family $\cF_{\cQ_{\mathrm{DL}_{k,m}}} = \mathrm{DL}_{k,m}$ on all of $\bits^m$.
\end{lemma}

We obtain Theorem \ref{thm:main3} of the introduction by combining Theorems~\ref{thm:releasepolys} and~\ref{dlfact}.  
\ifnum\icalp=0

\section{Generalizations and Limitations of Our Approach}
We note that the approach we take is not specific to low-degree polynomials.  Theorem~\ref{thm:releasepolys} extends to the case where the family $\cF_{\cQ}$ is efficiently and uniformly approximated on $Y_{\cQ}$ by linear combinations of functions from any efficiently computable set $\mathcal{S}$ (In the case of polynomials, $\mathcal{S}$ is the set of all monomials of total degree at most $t$).  The properties we require from the function family $\mathcal{S}$ are that it (1) it is relatively small (as it determines our running time and minimum database size), (2) for every $x \in \cX$, we can efficiently compute $\vec{b}_{x} \in \R^{\mathcal{S}}$ such that $\max_{y \in Y_{\cQ}} |f_{x}(y) - \sum_{s \in \mathcal{S}} \vec{b}_{x, s} \cdot s(y)| \leq \gamma$, and (3) $s(y)$ is small (say, at most $c$) for every $s \in \mathcal{S}$ and $y \in Y_{\cQ}$.  In the special case of approximation by $m$-variate real polynomials of degree $t$ and norm $T$, we can take $\mathcal{S}$ to be the set of monomials of total degree $t$, thus $|\mathcal{S}| = \binom{m+t}{t}$ and $c = 1$.  If we have those parameters, then similarly to Theorem~\ref{thm:releasepolys} we can obtain an $\eps$-differentially private sanitizer with running time $\poly(n, \log|\cX|, |\mathcal{S}|, \log T, \log(1/\eps))$ and accuracy $\alpha = \gamma + O(c T |\mathcal{S}|^2 \log(|\mathcal{S}|/\beta)/\eps n)$.

Unfortunately, it turns out that, for all of the query sets considered in this paper, there is not much to be gained from considering more general functions families $\mathcal{S}$. Indeed, Klivans and Sherstov \cite[Theorem 1.1]{KlivansSh10} show that even
if $\mathcal{S}$ consists of \emph{arbitrary} functions $s_1, \dots, s_{\ell}: \{0, 1\}^k \rightarrow \mathbb{R}$
whose linear combinations can uniformly approximate every monotone disjunction on $k$ variables to error $\pm 1/3$, then $\ell \geq 2^{\Omega(\sqrt{k})}$. Note that, up to logarithmic factors, this matches the dependence on $k$ of the upper bound of Lemma \ref{thm:approxdisj}. 
Moreover, Sherstov \cite[Theorem 8.1]{Sh11pattern} broadly extends the result of \cite{KlivansSh10} beyond disjunctions, to \emph{pattern matrices} of any Boolean function $f$
with high \emph{approximate degree}. Roughly speaking, the pattern matrix of $f$ corresponds to the query set consisting of all restrictions of $f$, with some variables possibly negated (see 
\cite{Sh11pattern} for a precise definition). In particular, \cite[Theorem 8.1]{Sh11pattern} implies a function family independent lower bound for non-monotone $r$-of-$k$ queries;
up to logarithmic factors, this lower bound matches the dependence on $r$ and $k$ of the upper bound of Lemma \ref{thm:approx2}.

We also note that our algorithm can be implemented in Kearns' \emph{statistical queries model}.  In the statistical queries model, algorithms can only access the database through a \emph{statistical queries oracle} $\mathrm{STAT}(D, \tau)$ that takes as input a predicate $q \from \bits^d \to [0,1]$ and returns a value $a$ such that $\Ex{x \in D}{q(x)} \in [a \pm \tau]$.  It can be verified that our algorithm can be implemented using $\binom{m+t}{t}$ queries to $\mathrm{STAT}(D, \tau)$ for $\tau = 1/\binom{m+t}{t} T^2$, using one query to obtain each noisy coefficient (scaled to the range $[0,1]$).  In this case of $k$-literal conjunctions, our algorithm makes $d^{O(\sqrt{k}\log(1/\alpha))}$ queries and requires $\tau = 1/d^{O(\sqrt{k}\log(1/\alpha))}$.

Connections between differentially private data analysis and the statistical queries model have been studied extensively~\cite{BlumDwMcNi05, KasiviswanathanLeNiRaSm07, GuptaHaRoUl11}.  Gupta et.~al.~showed that lower bounds for \emph{agnostic learning} in the statistical queries model also imply lower bounds for data release in the statistical queries model.  Specifically, applying a result of Feldman~\cite{Feldman09}, they show that for every $\alpha = \alpha(d) = o(1)$, for $k = k(d) = \Theta(\log(1/\alpha(d))) = \omega(1)$, and for every constant $c > 0$, there is no algorithm that makes $d^c$ queries to $\mathrm{STAT}(D, d^{-c})$, for $D \in (\bits^d)^*$, and releases $\alpha$-accurate answers to all monotone $k$-way disjunction queries on $D$.  Notice that our algorithm only makes a polynomial number of statistical queries when \emph{both} $k$ and $\alpha$ are constant.

\else
\fi

\section*{Acknowledgements}

We thank Vitaly Feldman, Moritz Hardt, Varun Kanade, Aaron Roth, Guy Rothblum, and Li-Yang Tan for helpful discussions.

\bibliographystyle{splncs03}
\bibliography{./privacyrefs}
\appendix

\ifnum\icalp=0
\section{Polynomial Approximation of Decision Lists}
\label{app:dlfact}
\begin{lemma}[Theorem~\ref{dlfact} restated, \cite{KlivansSe04}]  For every $k, m \in \N$ such that $k \leq m$ and every $\gamma > 0$, the family $\cP_{t, T}$ of $m$-variate real polynomials of degree $\tilde{O}\left(\sqrt{k} \log(1/\gamma)\right)$ and norm $T = m^{\tilde{O}\left(\sqrt{k} \log(1/\gamma)\right)}$ efficiently and uniformly $\gamma$-approximates the family $\cF_{\cQ_{\mathrm{DL}_{k,m}}} = \mathrm{DL}_{k,m}$ on all of $\bits^m$.
\end{lemma}

\begin{proof} By Theorem \ref{thm:releasepolys}, it is sufficient to show that if $f(y)$ is any length-$k$ decision list over $\bits^m$, then $f(y)$ can be $\gamma$-approximated 
by an explicit family of polynomials of degree $\tilde{O}(\sqrt{k} \log(1/\gamma))$ and norm $m^{\tilde{O}(\sqrt{k} \log(1/\gamma))}$. To this end, write $f(y)$ in the form,
``if $\ell_1$ then output $b_1$ else $\cdots$ else if $\ell_k$
then output $b_k$ else output $b_{k+1}$,'' where each
$\ell_i$ is a boolean literal, and each $b_i$ is an output
bit in $\{0,1\}.$ Assume for notational convenience that $\ell_i=y_i$ for all $i$; the proof for general decision lists is similar.

Following \cite[Theorem 8]{KlivansSe04}, we may write
\begin{align*}
f(y) 
&= b_1y_1 + b_2(1-y_1)y_2 \\
&+ \dots \\
&+ b_k(1-y_1) \dots(1-y_{k-1})y_k + b_{k+1}(1-y_1)(1-y_2)\dots(1-y_k).
\end{align*}

At a high level, we treat each term of the above sum independently, using a transformation of the Chebyshev polynomials to approximate each term
within additive error $\gamma/k$. This ensures that 
that the sum of the resulting polynomials approximates
$f_{\cQ, x}(y)$ within additive error $\gamma$ as desired. Details follow.

Let $g_k$ be the polynomial described in Lemma \ref{lemma:chebyshev} with error parameter $\gamma'=\gamma / k$. Then the polynomial $h_k(z)=1-g_k(k-z)$ satisfies the following properties:

\begin{enumerate}
\item The degree of $h_k$ is $t_k = O(\sqrt{k} \log(k/\gamma))$,
\item  for every $i \in \set{0,1,\dots, t_k}$, the $i$-th coefficient, $c_i$, of $h_k$ has magnitude $|c_i| \leq  2^{\tilde{O}(\sqrt{k} \log(1/\gamma))}$,
\item  $h_k(k) = 1$, and
\item for every $z \in \set{0, \dots, k-1}$, $|h_k(z)| \leq \frac{\gamma}{k}$.
\end{enumerate}
Moreover, $h_k$ can be constructed in time $\poly(k, \log(1/\gamma))$ (e.g.~using linear programming, though faster algorithms for constructing $h_k$ are known).

Consider the polynomial $p_x$ defined as 

\[\begin{array}{lll}
p_x(y) & = & b_1 \cdot h_k\left(y_1 + (k-1)\right) + b_2 \cdot h_k\left((1-y_1)+y_2 + (k-2)\right) + \dots  \\
& + & b_k \cdot h_k\left((1-y_1) + (1-y_2) + \dots + (1-y_{k-1}) + y_k\right) \\
  & + & b_{k+1} \cdot h_k\left((1-y_1) + (1-y_2) + \dots + (1-y_k)\right).\end{array}\]

It is easily seen that $p_x$ is has degree $O(\sqrt{k} \log(k/\gamma))$ and the absolute value of each of the coefficients is at most $2^{\tilde{O}(\sqrt{k} \log(1/\gamma))}$. Moreover, $|p_x(y) - f(y)| \leq \gamma$ for all $y \in \{0, 1\}^m$. This completes the proof.

\section{Polynomial Approximation of $r$-of-$k$ Queries}
\label{sherstovapp}
\begin{lemma}[Lemma~\ref{lemma:paturi} restated,\cite{Sherstov09}]
For every $r,k \in \N$ such that $r \leq k$ and $\gamma > 0$, there exists a univariate polynomial $g_{r,k}\from \R \to \R$ of degree $t_{r,k}$ such that $g_{r,k}(x) = \sum_{i=0}^{t_k} c_i x^i$ and:
\begin{enumerate}
\item $t_{r,k} = O\left(\sqrt{rk} \log(k) + \sqrt{k \log(1/\gamma)\log(k)}\right)$ ,
\item for every $i \in \set{0,1,\dots, t_k}, |c_i| \leq  2^{\tilde{O}(\sqrt{kr \log(1/\gamma)})}$,
\item for every $x \in \set{0,1,\dots,r-1}$, $-\gamma \leq g_{r,k}(x) \leq \gamma$, and
\item for every $x \in \set{r, \dots, k}$, $1-\gamma \leq g_{r,k}(x) \leq 1+\gamma$.
\end{enumerate}
Moreover, $g_{r,k}$ can be constructed in time $\poly(k, r, \log(1/\gamma))$ (e.g.~using linear programming).
\end{lemma}
We give details for the construction of \cite[Lemma 3.11]{Sherstov09}. We do not prove the approximation properties (which is done in~\cite{Sherstov09}), but just confirm the sizes of the coefficients.  
The result is trivial if $r=\tilde{\Omega}(k)$ or $\log(1 / \gamma) = \tilde{\Omega}(k)$ since every symmetric function on $\set{0,1,\dots,k}$
has an exact representation as a polynomial of degree $k$, so assume $r < k / \log^2 k$ and $\log(1/\gamma) \leq k / \log k$, with $k \geq 2$. 


Let $T_z$ be the degree $z$ Chebyshev polynomial of the first kind (Fact~\ref{fact:coeff}). We will use the following well-known properties of Chebyshev polynomials.

\begin{fact} \label{fact:coeff} The Chebyshev polynomials of the first kind satisfy the following properties.

\begin{enumerate}
\item Each coefficient of $T_z$ has absolute value at most $3^z$.
\item $T_z(t) \geq 1$ for all $t \geq 1$. 
\end{enumerate}
\end{fact} 

Let $\Delta = \big\lceil \frac{\log(k / \gamma)}{\log n}\big\rceil$, and $z=3\Delta\big\lceil \log k\big\rceil$. 
The construction proceeds in several steps, with the final polynomial $p$ defined in terms of multiple intermediate polynomials. 

\begin{enumerate}
\item For any fixed integer $0<\ell<r$,  let 
$$p_{1, \ell}(t) = T_{\big\lceil\sqrt{\frac{k-\ell-\Delta}{\ell+\Delta}}\big\rceil}\left(\frac{t}{k-\ell-\Delta}\right).$$ 
All the coefficients of $p_{1, \ell}(t)$ are bounded in absolute value by $3^{\text{deg}(p_{1, \ell})} = 2^{\tilde{O}\left(\sqrt{\frac{k-\ell-\Delta}{\ell+\Delta}}\right)}$.  

\item Define $$p_{2, \ell}(t) = \left(\frac{p_{1, \ell}(t)-p_{1, \ell}(k-\ell)}{8}\right)^2.$$ The coefficients of $p_{2, \ell}$ are bounded in absolute value by 
$$k^{O(\deg(p_{2, \ell}))}= k^{O(\deg(p_{1, \ell}))} = 2^{\tilde{O}\left(\sqrt{\frac{k-\ell-\Delta}{\ell+\Delta}}\right)}.$$

\item Define the polynomial $$p_{3, \ell}(t):=T_{\big\lceil\frac{8(d+1)(\ell+\Delta)}{\sqrt{2}\Delta}\big\rceil}\left(1 + \frac{\Delta^2}{64(\ell + \Delta)^2}-p_2(t)\right).$$

$p_{3, \ell}(t)$ has degree at most 
\begin{align*}
\deg(p_{3, \ell})
&={} 22(d+1)\sqrt{k(\ell+\Delta)}/\Delta \\
&={} O\left(\sqrt{k(\ell+\Delta)} \log k\right) = O\left(\sqrt{k \ell}\log k + \sqrt{k \log k \log(1/\gamma)}\right).
\end{align*}
Noting that $\frac{\Delta^2}{64(\ell + \Delta)^2} < 1$, it is clear that the coefficients of $p_{3, \ell}(t)$ are bounded in absolute value by $k^{O(\deg(p_{3, \ell}))}$.

\item Define the polynomial $p_{4, \ell}(t)=\frac{p_{3, \ell}(t)}{p_{3, \ell} (k-\ell)}$. Since $$p_3(k-\ell) = T_{\big\lceil{\frac{8(d+1)(\ell+\Delta)}{\sqrt{2}\Delta}\big\rceil}}\left(1 + \frac{\Delta^2}{64(\ell + \Delta)^2}-0\right),$$ by Part 2 of Fact \ref{fact:coeff}, the absolute values of the coefficients of $p_{4, \ell}$ are no larger than those of $p_3$.

\item Define the univariate polynomial $$q_1(t)=\prod_{i=-(\Delta-1),\dots,(\Delta-1), i \neq 0} (t-(k-\ell-i)).$$ 

The coefficients of $q_1(t)$ have absolute value at most $k^{O(\Delta)}$.

\item Define the polynomial 

$$ q_{2, \ell}(t) = \frac{1}{q_1(k-\ell)} \cdot p_{4, \ell}(t)q_1(t).$$

Noticing that $|q_1(k-\ell)| > 1$, it is clear that the degree of $q_{2, \ell}$ is $O(\Delta + \deg(p_{4, \ell}))$, and the absolute values of the coefficients of $q_{2, \ell}(t)$ are at most
$2^{\tilde{O}(\Delta + \deg(p_{3, \ell}))}$. 

\item Define the polynomial $q_{3, \ell}(t) = q_{2, \ell}(k-t)$. The absolute value of the coefficients of $q_{3, \ell}$ are also 
bounded above by $2^{\tilde{O}(\Delta + \deg(p_{3, \ell}))}$. 

Sherstov's arguments show that $q_{3, \ell}$ is a $\frac{\gamma}{k}$-approximation to the function $$\text{EXACT}_{\ell}(t) = \begin{cases} 1 \text{ if } t = \ell\\ 0 \text{ for all other } t \in \{0, \dots, k\}\end{cases}.$$

\item The final polynomial $p(t)$ is defined as 
$$p(t)=1-\sum_{\ell \in \{0, 1, \dots, r-1\}} q_{3, \ell}(t).$$
The coefficients of $p(t)$ have absolute value at most $2^{\tilde{O}(\Delta + \deg(p_{3, r}))}$, and $p$ has the desired degree.
Moreover, $p(t)$ $\gamma$-approximates the function $\mathbf{1}_{t \geq r}(t)$ on the set $\{0, \dots, k\}$, since the $j$th term in the sum 
$\sum_{\ell \in \{0, 1, \dots, r-1\}} q_{3, \ell}(t)$ is a  $\gamma / k$ approximation for the function
$\text{EXACT}_{j}(t)$.

Since $\log(1/\gamma) < \frac{k}{\log k}$, it holds that $\Delta < \sqrt{k \log 1/\gamma}$. Thus, the coefficients of $p(t)$ have absolute value at most $2^{\tilde{O}(\deg(p_{3, r}))} = 2^{\tilde{O}(\sqrt{k r} + \sqrt{k \log(1/\gamma)})}$.

\end{enumerate}

\end{proof}
\else
\fi

\end{document}